\documentclass[conference,letterpaper]{IEEEtran}

\newif\ifarxiv
\arxivtrue 

\usepackage[utf8]{inputenc} 
\usepackage[T1]{fontenc}
\usepackage{url}
\usepackage{ifthen}
\usepackage{cite}
\usepackage{graphicx} 
\usepackage{subfigure}
\usepackage{stfloats}
\usepackage[cmex10]{amsmath} 

\usepackage{amssymb,amsfonts} 
\usepackage{mathtools}
\usepackage[normalem]{ulem} 
\usepackage{algorithm}
\usepackage[noend]{algpseudocode}
\usepackage{enumitem}
\usepackage{xcolor}
\usepackage{nicefrac}

\IEEEoverridecommandlockouts

\newtheorem{lemma}{Lemma}
\newtheorem{theorem}{Theorem}
\newtheorem{proposition}{Proposition}

\newtheorem{remark}{Remark}
\newtheorem{corollary}{Corollary}
\newtheorem{proof}{Proof}

\begin{document}
\title{Best-Arm Identification with Noisy Actuation} 


\author{\IEEEauthorblockN{Merve Karakas$^\dagger$, Osama Hanna$^\ddagger$, Lin F. Yang$^\dagger$, and Christina Fragouli$^\dagger$\\ 
$^\dagger$University of California, Los Angeles, $^\ddagger$ Meta, Central Applied Science\\
Email: mervekarakas@ucla.edu, ohanna@meta.com, \{linyang, christina.fragouli\}@ucla.edu}\thanks{This work is partially supported by NSF grants \#2221871, \#2007714, and \#2221871, by Army Research Laboratory grant under Cooperative Agreement W911NF-17-2-0196, and by Amazon Faculty Award.}
}

\maketitle


\begin{abstract}
In this paper, we consider a multi-armed bandit (MAB) instance and study how to identify the best arm when arm commands are conveyed from a central learner to a distributed agent over a discrete memoryless channel (DMC). 
Depending on the agent capabilities, we provide communication schemes along with their analysis, which interestingly relate to the zero-error capacity of the underlying DMC.
\end{abstract}

\section{Introduction}
\label{sec:intro}
Motivated by the growing interest in distributed learning applications, recent work has begun to investigate learning over noisy communication channels \cite{singleagentprior,hanna2024multi, karakas2025feedback, karakas2026jsait, bai-arm-erasures}. Within this landscape, we consider a new formulation: best arm identification in which arm commands are conveyed from a central learner to a distributed agent over a discrete memoryless channel (DMC).

Such channels can model situations in which actions are communicated over unreliable interfaces, including physical controls, low-bandwidth links, or human-mediated instructions, where commands may be received correctly or confused with a small set of alternatives. Despite their simplicity, DMCs capture key sources of noise in practical action-communication pipelines and allow us to derive theoretical insights.

{Our focus is on identifying when performance guarantees can be made independent of the channel error probabilities and depend only on confusability. For that question, zero-error capacity is the natural threshold quantity.}

In this paper, we study three models of increasing capability for the distributed agent. In the first model, the agent can only execute the commands it receives. In the second, the agent is equipped with a preloaded codebook that allows it to decode the received signals into actions. In the third, most powerful model, the agent {can maintain state and execute multi-round plans installed via zero-error packets}.
For each model, we compare performance against an idealized benchmark in which communication occurs over an error-free channel. Our main findings are summarized below:

\begin{itemize}[leftmargin=*]
\item \textbf{Case 1 (No decoding).} If the agent simply executes received commands, the channel induces mixing of the action distributions. Performance degrades by a factor governed by the smallest singular value of the channel matrix, which depends on the channel error probability.
\item \textbf{Case 2 (Fixed decoding).} If the zero-error capacity of the channel is nonzero and the agent can use a fixed block code, then the resulting performance loss can be only a constant multiplicative factor, independent of the channel error probabilities.
\item \textbf{Case 3 ({Stateful execution}).} If the agent {can maintain state and execute multi-round plans installed via zero-error packets,} then the performance gap relative to the error-free benchmark can be reduced to an additive overhead.
\end{itemize}
We note that when the zero-error capacity of the underlying DMC is zero, no coding strategy can eliminate dependence on the channel noise (Remark~\ref{rem:C0-zero}), and performance necessarily degrades with the channel error probability. 

\smallskip
\noindent \textbf{Related Work.} Multi-armed bandits (MAB) are a standard model for sequential decision-making under uncertainty; see, e.g.,
\cite{lattimore2020bandit} and references therein. Our objective is fixed-confidence best-arm identification (BAI),
for which classic elimination and tracking-style algorithms yield instance-dependent sample complexities of order
$\tilde O\!\big(\sum_{a\neq a^\star}\Delta_a^{-2}\log(1/\delta)\big)$; representative references include
successive elimination \cite{even2006action}, lil'UCB \cite{jamesnowak2014}, and Track-and-Stop
\cite{garivier2016optimal} as well as complexity/lower-bound characterizations \cite{kaufmann2016complexity}.

Prior work on noisy bandits focuses on the \emph{reward} channel---adversarial corruption \cite{Lykouris18, pmlr-v99-gupta19a, amir2020prediction} and delayed/censored feedback \cite{Kapoor19ProbCorReward}---but assumes the chosen arm is executed as intended. Recent work has studied \emph{uncertainty in the executed arm} via arm-erasure models, for regret \cite{singleagentprior,hanna2024multi, karakas2025feedback, karakas2026jsait} and BAI \cite{bai-arm-erasures}. In contrast, we consider general discrete memoryless actuation channels with \emph{confusability} (typewriter channels as a running example), leveraging zero-error communication tools \cite{Shannon1956,lovasz1979shannon,KornerOrlitsky1998ZeroError}.

\smallskip
\noindent \textbf{Paper Organization.} Section~\ref{sec:problem} provides our model, objectives, and background on zero-error capacity; Sections~\ref{sec:baseline}, \ref{sec:typewriter-anchors}, and~\ref{sec:PSE} discuss cases 1--3 respectively.\ifarxiv\else\ Some proof details are omitted for space; see~\cite{driveversion}.\fi

\section{Model and Objectives}
\label{sec:problem}

\subsection{Bandit instance and objective (BAI)}
We consider a stochastic $K$-armed bandit with arms indexed by $\mathcal{K} \coloneq \{0,1,\dots,K-1\}$.
Pulling arm $a$ produces a reward $r\sim\nu_a$ with mean $\mu_a\triangleq\mathbb{E}[r]$.
Let $a^\star\in\arg\max_{a\in \mathcal{K}}\mu_a$ denote a best arm (assumed unique for simplicity), and define gaps
$\Delta_a\triangleq \mu_{a^\star}-\mu_a$ for $a\neq a^\star$.

Our goal is \emph{fixed-confidence best-arm identification (BAI)}:
an algorithm adaptively interacts with the environment, stops at a (random) time $\tau$, and outputs $\hat a$.
It is $\delta$-correct if
\[
\Pr(\hat a = a^\star)\ge 1-\delta.
\]
Hence, we measure performance primarily by the total number of physical rounds (pulls).
In our interface, each round also uses the command channel once, so the pull count
coincides with the number of channel uses.

In the noiseless setting, instance‑dependent lower and upper bounds show that fixed‑confidence BAI can be solved with $ \Theta \left( \sum_{a \neq a^\star} \Delta^{-2}_{a} \log{(1/\delta)}\right)$ samples (see, e.g., \cite{jamesnowak2014, even2006action, garivier2016optimal}). We will treat this benchmark as $N_{\text{clean}}(\delta,\mu)$.


\subsection{Actuation Channel and Arm Mismatch}
\label{sec:channel-model}

At each time slot $t$, the learner produces a channel input $X_t\in\mathcal{X}$ which is passed through a discrete
memoryless channel (DMC) $W:\mathcal{X}\to\mathcal{Y}$, yielding an output $Y_t\in\mathcal{Y}$ at the agent according to
\[
W(y\mid x)\;\triangleq\;\Pr(Y_t=y\mid X_t=x).
\]
We consider an \emph{arm channel} with $\mathcal X=\mathcal Y=\mathcal K$, so $X_t,Y_t\in\mathcal K$ and
block codes correspond to sequences $X^n,Y^n\in\mathcal K^n$.
The agent then executes an arm $\tilde{a}_t\in \mathcal{K}$ (as a function of its received symbols and actuation rule), and a reward
$r_t$ is generated with mean $\mu_{\tilde{a}_t}$. The learner observes $r_t$ and its own transmissions $X_t$, but does not observe channel outputs 
$Y_t$.

\subsection{Zero-error capacity and confusability graphs}
\label{sec:zero-error-prelim}
We briefly introduce notation and review zero-error communication over a DMC.
Zero-error rates depend only on which outputs are \emph{possible} for each input (the support of $W$),
not on their probabilities \cite{Shannon1956}. Define the output support sets
\[
\mathcal{Y}(x)\triangleq \{y\in\mathcal{Y}: W(y\mid x)>0\},\qquad x\in\mathcal{X}.
\]
\paragraph{Confusability graph}
The confusability graph of $W$ is the undirected graph $G=(V,E)$ with $V=\mathcal{X}$ and
\[
\{x,x'\}\in E
\quad\Longleftrightarrow\quad
\mathcal{Y}(x)\cap \mathcal{Y}(x')\neq \emptyset.
\]
\paragraph{Blocklength-$n$ packets}
A length-$n$ packet is an input sequence $x^n\in V^n$.
Confusability of length-$n$ sequences is captured by the $n$-fold strong graph power $G^{\boxtimes n}$:
distinct $x^n,(x')^n\in V^n$ are adjacent in $G^{\boxtimes n}$ iff for all coordinates $i$,
$x_i=x'_i$ or $\{x_i,x'_i\}\in E$ \cite{Shannon1956, KornerOrlitsky1998ZeroError}.
\paragraph{Zero-error message count}
Let $\alpha(\cdot)$ be the independence number. The maximum number of messages conveyable with zero decoding error
using exactly $n$ channel uses is
\begin{equation}
\label{eq:Mn-alpha}
M(n)\;\triangleq\;\alpha\!\big(G^{\boxtimes n}\big),
\end{equation}
via the standard equivalence between zero-error codebooks and independent sets \cite{Shannon1956}.
\paragraph{Minimal blocklength for a message set}
For a finite message set $\mathcal{S}$ with $|\mathcal{S}|<\infty$, define
\begin{equation}
\label{eq:nstar-def}
n^\star(\mathcal{S}) \;\triangleq\;
\min\{n\in\mathbb{N}:\alpha(G^{\boxtimes n})\ge |\mathcal{S}|\}.
\end{equation}
\paragraph{Zero-error capacity}
Shannon's zero-error capacity is
\begin{equation}
\label{eq:C0-def}
C_0(G)\;\triangleq\;\lim_{n\to\infty}\frac{1}{n}\log_2 \alpha\!\big(G^{\boxtimes n}\big)
\;=\;\log_2 \Theta(G),
\end{equation}
where $\Theta(G)$ is the Shannon capacity \cite{Shannon1956}.
When $C_0(G)>0$, $M(n)$ grows exponentially in $n$, so any fixed finite message set is achievable
with some finite blocklength. 

\begin{remark}
\label{rem:C0-zero}
The following are equivalent (see \cite{Shannon1956}):
\[
C_0(G)=0
\;\Longleftrightarrow\;
\alpha(G^{\boxtimes n})=1\ \forall n
\;\Longleftrightarrow\;
G \text{ is complete}.
\]
Here, \(M(n)=1 ~\forall ~n\), so no fixed-blocklength, no-feedback protocol can convey even a
binary command with zero-error. 
\end{remark}

\subsection{Typewriter channel and general graphs}
\label{sec:typewriter}
A central example is the one-sided typewriter channel (see Fig.~\ref{fig:typewriter-cycle}) over alphabet $\mathcal{X}=\{0,1,\dots,K-1\}$:
\[
Y=\begin{cases}
X, & \text{w.p. }1-\varepsilon,\\
X+1\!\!\!\!\pmod K, & \text{w.p. }\varepsilon.
\end{cases}
\]
More generally, our coding layer is described by the channel's \emph{confusability graph} $G$,
which allows us to state results for arbitrary discrete actuation links beyond typewriter channels.
For any $\varepsilon\in(0,1)$, the one-sided typewriter channel has confusability graph
$G=C_K$, the undirected cycle on $K$ vertices with edges $\{i,i+1\}$ (indices modulo $K$); throughout the paper we use $C_5$ and $C_6$ as running examples.


\begin{figure}[t b!]
  \centering
  \includegraphics[width=0.8\linewidth]{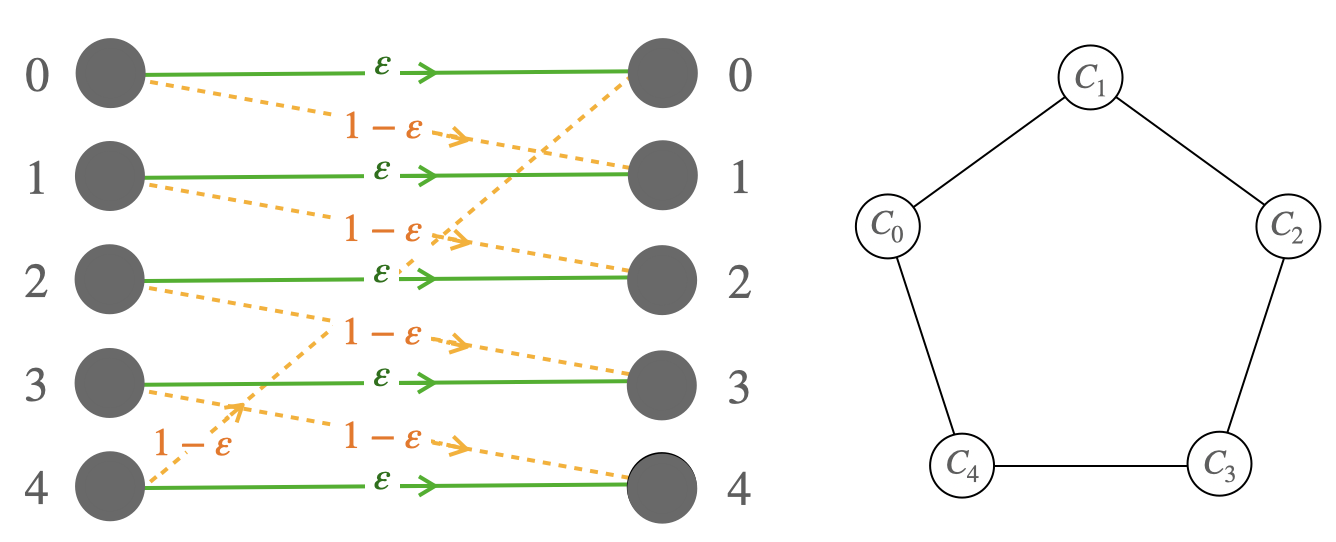}\
  \caption{Example of a one-sided typewriter channel over alphabet $\mathcal{X}=\{0,\dots,4\}$ (left) and its confusability graph $C_5$ (right)}.
  \vspace{-.2in}
  \label{fig:typewriter-cycle}
\end{figure}

\section{Case 1: No decoding and $\varepsilon$-dependent inflation}
\label{sec:baseline}
We begin with the vanilla single-shot method: in each round $t$, the learner selects an intended arm
$a_t\in \mathcal{K}$ and transmits it once over the channel. The agent then executes the received symbol
$\tilde{a}_t$ as the pulled arm immediately, generating a reward $r_t$ with mean $\mu_{\tilde{a}_t}$. The learner observes $(a_t,r_t)$
but not $\tilde{a}_t$.

Let $\mu=(\mu_0,\dots,\mu_{K-1})^\top$ be the (unknown) physical arm-mean vector. Define the
\emph{command-conditioned} mean vector ${\tilde\mu}=({\tilde\mu_0},\dots,{\tilde\mu_{K-1}})^\top$ by $
{\tilde\mu_i} \triangleq \mathbb{E}[r_t\mid a_t=i],\quad i \in \mathcal{K}.
$

For a general DMC with transition matrix $W$, the single-shot method yields
\begin{equation}
\label{eq:baseline-nu-Wmu}
{\tilde\mu_i} \;=\; \sum_{y=0}^{K-1} W(y\mid i)\,\mu_y,
\qquad\text{i.e.}\qquad
{\tilde\mu} = W\mu.
\end{equation}
Thus, in the baseline, the learner observes rewards from a \emph{mixed instance} ${\tilde\mu}$, not directly from $\mu$.

\paragraph{Typewriter specialization.}
Under the typewriter law from Section~\ref{sec:typewriter},
we have
\begin{equation}
\label{eq:baseline-typewriter-mixing}
{\tilde\mu_i} \;=\; (1-\varepsilon)\mu_i + \varepsilon\,\mu_{i+1}, \qquad i \!\!\!\!\pmod K.
\end{equation}
i.e., $W=(1-\varepsilon)I+\varepsilon S$, where $S$ is the cyclic shift $(S\mu)_i=\mu_{i+1}$.

\subsubsection{Identifiability: when the best arm cannot be recovered}
\label{sec:baseline-identifiability}

The BAI goal is to identify $a^\star\in\arg\max_i \mu_i$.
This is impossible if the mixing map $\mu\mapsto {\tilde\mu}=W\mu$ is not injective.

\begin{lemma}[Non-identifiability under non-injective mixing]
\label{lem:baseline-nonident}
If there exist $\mu\neq \mu'$ such that $W\mu=W\mu'$ but $\arg\max_i \mu_i\neq \arg\max_i \mu'_i$, then no
algorithm can be $\delta$-correct for all instances (for any $\delta<1$).
\end{lemma}

\subsubsection{Baseline inflation via conditioning of the mixing map}
\label{sec:baseline-conditioning}

Under single-shot mixing \eqref{eq:baseline-nu-Wmu}, the learner observes the mixed means
${\tilde\mu}=W\mu$ rather than $\mu$. When $W$ is invertible, a natural baseline is to estimate ${\tilde\mu}$
from samples and then \emph{unmix} via $\hat\mu = W^{-1}{\hat{\tilde\mu}}$. This unmixing step amplifies
estimation error according to the conditioning of $W$.

\begin{proposition}
\label{prop:baseline-conditioning}
Assume $W$ is known and invertible. Let $\sigma_{\min}(W)$ denote the smallest singular value of $W$.
Then, relative to the noiseless setting, the single-shot + unmixing baseline incurs an
inflation governed by
\begin{equation}
\begin{aligned}
\label{eq:baseline-inflation}
T_{\text{baseline}}(\delta,\mu)
&\;=\;
\tilde{O}\!\Big(\frac{1}{\sigma_{\min}(W)^2}\,N_{\text{clean}}(\delta,\mu)\Big), \\
\end{aligned}
\end{equation}
where $\tilde{O}(\cdot)$ hides universal constants and logarithmic factors and \begin{equation}
\label{eq:Nclean-recall}
N_{\text{clean}}(\delta,\mu)
\;=\;
\Theta\!\Big(\sum_{a\neq a^\star} \Delta_a^{-2}\log(1/\delta)\Big).
\end{equation}
\end{proposition}

\emph{Proof sketch.}
Let ${\hat{\tilde\mu}}$ be an estimator of ${\tilde\mu}$ and define $\hat\mu=W^{-1}{\hat{\tilde\mu}}$.
By submultiplicativity,
\begin{equation}
\|\hat\mu-\mu\|_2
=
\|W^{-1}({\hat{\tilde\mu}}-{\tilde\mu})\|_2
\le
\frac{1}{\sigma_{\min}(W)}\,\|{\hat{\tilde\mu}}-{\tilde\mu}\|_2.
\end{equation}
In the noiseless (direct-actuation) setting, standard fixed-confidence BAI algorithms satisfy \eqref{eq:Nclean-recall} pulls. Thus, to achieve the same accuracy on $\mu$ as in the noiseless case, it suffices to estimate
${\tilde\mu}$ more accurately by a factor $\sigma_{\min}(W)$.
Since statistical estimation error scales as $t^{-1/2}$, shrinking the target error by a factor
$\sigma_{\min}(W)$ increases the required sample size by a factor $1/\sigma_{\min}(W)^2$,
yielding \eqref{eq:baseline-inflation}.{\footnote{{Note that $\sigma_{\min}(W)\le 1$ for any square stochastic matrix, so $1/\sigma_{\min}(W)^2\ge 1$ and the baseline never improves over the noiseless case.}}}

\begin{remark}
  In general, the baseline inflation is governed by the conditioning of the mixing map via
$1/\sigma_{\min}(W)^2$. For the typewriter mixing matrix $W=(1-\varepsilon)I+\varepsilon S$, when $K$ is even
we have $\sigma_{\min}(W)=|1-2\varepsilon|$, so the inflation reduces to $(1-2\varepsilon)^{-2}$ and diverges
as $\varepsilon\to\tfrac12$, where identifiability fails.  
\end{remark}

\noindent\emph{Example ($C_5$,$C_6$):} 
\begin{itemize}[leftmargin=*]
    \item K=6 (even): $\sigma_{\min}(W)=|1-2\varepsilon|$ so baseline inflation blows up as $\varepsilon\to\tfrac12$ and is impossible at $\varepsilon=1/2$.
    \item K=5 (odd): $\sigma_{\min}(W)>0$ for all $\varepsilon \in (0,1)$ so baseline remains identifiable (no singularity), but still suffers conditioning-based inflation.
\end{itemize}

\section{Case 2: Zero-Error Block Codes}
\label{sec:typewriter-anchors}
In this section, we consider systems that enable a learner and an agent to preshare a fixed block code that achieves zero-error message transmission in $n$ channel uses. This removes $\varepsilon$-dependent conditioning entirely (since correctness is zero-error and depends only on the channel support), but introduces a {constant multiplicative slowdown whose value depends on the zero-error control scheme.} We present two generic such \emph{zero-error} coding construction methods, using for illustration the typewriter channels  $C_5$ and $C_6$ (Section~\ref{sec:typewriter}).


\begin{itemize}[leftmargin=*]
    \item \textbf{Scheme~1: zero-error capacity codebook.} A length-$n_u$
    zero-error codebook for $K$ messages allows the learner to install any arm index
    once every $n_u$ channel uses.
    \item \textbf{Scheme~2: independent-set schedules.} The learner uses a public
    schedule over independent sets of the confusability graph; in each time slot,
    only arms in the active independent set may be transmitted, but each such arm
    can be decoded in a \emph{single} channel use.
\end{itemize}



\subsection{Scheme 1: zero-error capacity codebook for arm updates.}
We illustrate on $C_5$, producing a $(2,1)$ zero-error code, where $(i, j)$ indicates $i$ channel use(s) to install $j$ arm choice(s).
\label{subsec:C5-slope}
Set $K=5$. Consider the length-$2$ codebook
\begin{equation}
\label{eq:C5-slope}
\varphi(i)\;\triangleq\;\big(i,\;2i \!\!\!\!\pmod 5\big),\qquad i\in\{0,1,2,3,4\}.
\end{equation}
Under the typewriter law, transmitting $\varphi(i)$ yields
\[
Y_1\in\{i,i+1\},\qquad
Y_2\in\{2i,2i+1\}\pmod 5.
\]

\begin{lemma}
\label{lem:C5-slope}
The codebook $\{\varphi(i)\}_{i=0}^4\subseteq \{0,\dots,4\}^2$ is zero-error for the typewriter channel
(with $C_5$). 
Equivalently, $\{\varphi(i)\}$ is an independent set in $C_5^{\boxtimes 2}$,
so $\alpha(C_5^{\boxtimes 2})\ge 5$.
\end{lemma}

This code conveys $\log_2 5$ bits in $2$ uses, i.e., rate $\tfrac{1}{2}\log_2 5$ bits/use.
Since $C_0(C_5)=\tfrac{1}{2}\log_2 5$ \cite{lovasz1979shannon}, it is capacity-achieving\footnote{By \eqref{eq:Mn-alpha} and Sec.~\ref{sec:typewriter}, the number of zero-error messages in $n$ channel uses is $M(n)=\alpha(C_K^{\boxtimes n})$.} on $C_5$.

\begin{proposition}
\label{lem:update-wrapper-modelS}
Assume there exists a fixed-length zero-error update code of blocklength $n_u$ with
$M(n_u)=\alpha(G^{\boxtimes n_u})\ge K$. Assume the agent executes the decoded update
\emph{immediately upon decoding}, i.e., on the last channel use of the packet.\\
Let $\mathsf A$ be any $\delta$-correct clean BAI algorithm and $\tau_{\rm clean}$ its pull count.
Then there exists an algorithm $\widetilde{\mathsf A}$ over the actuation link that is $\delta$-correct and whose
total number of physical rounds satisfies
\[
\tau \;\le\; n_u\,\tau_{\rm clean}.
\]
\end{proposition}
\emph{Proof sketch.}
Run $\mathsf A$ as a virtual algorithm.
Partition time into consecutive blocks of length $n_u$.
In block $i$, transmit the $n_u$-symbol zero-error codeword for the arm requested by $\mathsf A$ at virtual step $i$.
By assumption, the decoded arm is executed on the last slot of the block; use the reward from that last slot
as the virtual reward fed back to $\mathsf A$.
All intermediate rewards within each block can be logged but ignored.
Zero decoding error ensures the virtual interaction matches the clean bandit interaction exactly.
Thus, $\delta$-correctness is preserved and each virtual pull costs $n_u$ physical rounds.\ifarxiv\ Full proof in App.~\ref{app:update-wrapper}.\fi
\begin{corollary}
If $\mathsf A$ satisfies $\tau_{\rm clean}\le N_{\rm clean}(\delta,\mu)$, then $\tau\le n_u\,N_{\rm clean}(\delta,\mu)$.
\end{corollary}

\begin{corollary}
\label{cor:nstar-56-modelA-modelS}
For $K=5$, Lemma~\ref{lem:C5-slope} gives an independent set of size $5$ in $C_5^{\boxtimes 2}$, so $M(2)\ge 5$ while $M(1)=\alpha(C_5)=2<5$. Hence $n^\star(5)=2$. For $K=6$, since $\alpha(C_6)=3$, we have $M(1)=3<6$ and by Cartesian-product coding, $M(2)\ge \alpha(C_6)^2 = 9 \ge 6$, hence $n^\star(6)=2$ for {Scheme~1} type of codes.
Therefore, a fully flexible zero-error update packet exists with blocklength $n_u=2$ for both $K=5$ and $K=6$. By Lemma~\ref{lem:update-wrapper-modelS}, the number of physical rounds satisfies $\tau\le 2\,N_{\rm clean}(\delta,\mu)$.
\newline
\indent These guarantees hold for all $\varepsilon\in(0,1)$, whereas, the single-shot baseline can be non-identifiable at $\varepsilon=\nicefrac{1}{2}$ for even $K$.
\end{corollary}

\begin{remark}[Non–zero-error preshared codes]
\label{rem:nonzero-error-codes}
Instead of \emph{zero}-error codes, one could preshare a classical block code (e.g., Reed–Solomon) with small decoding error.
From the learner’s viewpoint this is equivalent to replacing $W$ by an effective DMC $\widetilde W$ on $\mathcal K$, so the analysis reduces to the
single-shot baseline of Section~\ref{sec:baseline} with $\sigma_{\min}(\widetilde W)$
instead of $\sigma_{\min}(W)$.
\end{remark}

\subsection{Scheme 2: independent sets using a public parity schedule}
\label{subsec:C6-calendar}

This scheme exploits the graph structure of $G$ more directly: instead of
sending arbitrary arms at every use, we restrict each time slot to an \emph{independent
set} of vertices, which allows zero-error decoding in a single channel use.
We illustrate on $C_6$ by providing a constrained $(1,1)$ zero-error code on $C_6$ with a public parity schedule.

\begin{lemma}
\label{lem:indset-decoding}
Let $G=(V,E)$ be the confusability graph of $W$ and $\mathcal S\subseteq V$ an independent set.
If the encoder restricts to inputs $X\in\mathcal S$ and the decoder knows $\mathcal S$,
then $X$ can be decoded with zero-error from a single channel output $Y$.
\end{lemma}

Fix a finite collection of independent sets
$\mathcal S_1,\dots,\mathcal S_m\subseteq\mathcal K$, and a public periodic schedule
$s:\mathbb N\to[m]$ known to both learner and agent. At time $t$, only arms in the
active set $\mathcal S_{s(t)}$ may be transmitted; by Lemma~\ref{lem:indset-decoding},
any arm in $\mathcal S_{s(t)}$ can be sent and decoded with zero-error in one channel use.

To simulate a clean BAI algorithm $\mathsf A$, the learner runs $\mathsf A$ virtually
and, whenever $\mathsf A$ requests arm $a$, waits until the first time $t$ such that
$a\in\mathcal S_{s(t)}$, then transmits $a$ and counts the resulting reward as the
virtual reward. The slowdown factor depends on how frequently each arm (or pair of
arms) appears in the active sets. We now instantiate this on $C_6$.

Set $K=6$. The cycle $C_6$ admits a $2$-coloring with independent sets
$\quad \mathcal S_{\mathrm{even}}=\{0,2,4\},\quad \mathcal S_{\mathrm{odd}}=\{1,3,5\}.$

Fix a \emph{public, deterministic schedule} known to encoder/decoder that specifies which set is active
at each channel use, e.g., alternating $\mathcal S_{\mathrm{even}},\mathcal S_{\mathrm{odd}},\dots$.
In a use where $\mathcal S$ is active, the encoder is restricted to transmit $X\in\mathcal S$.

\begin{lemma}
\label{lem:oneuse-indset}
For the typewriter channel on $C_6$, if $X$ is restricted to an independent set
$\mathcal S\in\{\mathcal S_{\mathrm{even}},\mathcal S_{\mathrm{odd}}\}$ known to the decoder, then $X$ can be
decoded from a single output $Y$ with zero-error.
\end{lemma}

Each use conveys one of $|\mathcal S|=3$ possibilities with zero-error, a ternary digit per channel use,
where the active set $\mathcal S$ acts as shared side information. Moreover, for even cycles $C_{2m}$, $C_0(C_{2m})=\log_2 m$ \cite{lovasz1979shannon}; hence $C_0(C_6)=\log_2 3$. The above scheme then achieves $\log_2 3$ bits/use, and is capacity-achieving on $C_6$.

\begin{proposition}
\label{prop:C6-calendar-wrapper}
Consider Scheme~2 with $K=6$ (public alternating schedule
$\mathcal S_{\rm even}$, $\mathcal S_{\rm odd}$).
In each slot, the transmitted arm is decoded
with zero-error in one channel use (Lemma~\ref{lem:oneuse-indset}).
\newline
Let $\mathsf A$ be any $\delta$-correct clean BAI algorithm,
{let $\tau_{\rm clean}$ be its (random) number of pulls under no channel noise, and let
$a_1,\ldots,a_{\tau_{\rm clean}}$ denote the random sequence of arm
requests produced by $\mathsf A$ on the clean instance, and define the parity class
sequence}
\[
{p_i \triangleq
\begin{cases}
E, & a_i \in \mathcal S_{\rm even},\\
O, & a_i \in \mathcal S_{\rm odd}.
\end{cases}}
\]
{Then there exists an actuation wrapper that simulates $\mathsf A$ over the noisy
link with the same error probability and a stopping time $\tau$ satisfying}
\begin{equation}
\label{eq:C6-calendar-bound}
{\tau = \tau_{\rm clean} + \mathbf{1}_{\{p_1=O\}} + \textstyle\sum_{i=2}^{\tau_{\rm clean}} \mathbf{1}_{\{p_i=p_{i-1}\}} \le 2\tau_{\rm clean}.}
\end{equation}
\end{proposition}

\emph{Proof sketch.}
The wrapper runs $\mathsf A$ virtually and only counts rewards when the requested arm belongs to the active parity class; otherwise it waits one slot and then transmits the arm. Zero-error decoding ensures the counted interaction matches the clean bandit path.
{If the $i$-th request has the same parity as the $(i{-}1)$-th, the wrapper must wait one extra slot; otherwise the next slot already has the correct parity. Summing these waiting slots yields~\eqref{eq:C6-calendar-bound}.}\ifarxiv\ Full proof in App.~\ref{app:prop-C6-calendar-wrapper}.\fi

{Equation~\eqref{eq:C6-calendar-bound} shows that the slowdown under Scheme~2 depends on the local
pattern of parity requests, not merely on the total number of pulls from
each parity class. In particular,}
\[
{\frac{\tau}{\tau_{\rm clean}}
=
1 + \frac{\mathbf 1\{p_1 = O\}
+\sum_{i=2}^{\tau_{\rm clean}}\mathbf 1\{p_i = p_{i-1}\}}
{\tau_{\rm clean}}
\;\in\; [1,2].}
\]
{When the requested parity alternates often, the slowdown is close to $1$; when the clean algorithm makes long runs
within the same parity class, the slowdown approaches $2$.}


\subsubsection{Beyond bipartitions: overlapping independent sets}
\label{subsec:overlapping-indsets}

The parity schedule on $C_6$ is based on a partition of the arms into two independent
sets. There is no requirement, however, that the active independent sets form a
partition: an arm may belong to several independent sets and thus be admissible in
multiple time slots. This can be exploited to reduce the worst-case slowdown constant.

As a simple illustration, consider $C_5$ with vertex set $\{0,1,2,3,4\}$.
Independent sets include, for example,
$\{1,3\}$,~$\{2,4\}$,~$\{0,2\},$
and we may use a length-$3$ periodic schedule that activates these sets in turn.
Arm $2$ then belongs to \emph{two} of the three sets and is therefore available in
$2/3$ of the slots, rather than $1/3$. This illustrates that one can bias the schedule toward specific arms by letting them
appear in multiple independent sets. Intuitively, this
higher availability fraction can translate into a smaller slowdown factor when
simulating a clean algorithm, because the wrapper needs to wait less frequently for an
arm to become admissible.

Designing schedules (and collections of independent sets) that optimally balance the
availability of individual arms and pairs of arms is a combinatorial problem closely
related to fractional colorings and covering designs of the confusability graph. In
principle, such designs can push the worst-case slowdown below the simple factor-$2$
bound in Prop.~\ref{prop:C6-calendar-wrapper}. We leave a systematic exploration
of this direction to future work.

\section{Case 3: Stateful Plan Execution via Packetized Successive Elimination (PSE)}
\label{sec:PSE}

The previous section showed that if we insist on issuing a fresh (zero-error) arm command for every pull,
then any clean BAI algorithm can be simulated with a {constant \emph{multiplicative} slowdown for per-pull zero-error control.} We now show that if the agent can maintain state and execute a \emph{multi-round plan}, then the learner needs to communicate only at decision times (phase boundaries). This converts actuation cost into an \emph{additive} overhead.

Throughout this section, the system acts continuously: while an $n$-symbol command packet is being transmitted, the agent keeps executing its currently committed behavior, and the newly decoded plan takes effect only after the packet is decoded (switching latency).


\subsection{Packetized Successive Elimination (PSE)}
\label{subsec:PSE}

We now describe Algorithm \ref{alg:PSE} (PSE) which builds on phased version of Successive Elimination \cite{even2006action}. The coding layer only appears through the existence of a zero-error plan packet
of length $n_r$ that conveys any \emph{phase plan}; explicit constructions (e.g., on $C_5$ and $C_6$) are
plugged in later.

\begin{algorithm}[t]
\caption{Packetized Successive Elimination (PSE)}
\label{alg:PSE}
\begin{algorithmic}[1]
\Require Confidence $\delta\in(0,1)$; phase budgets $m_r=2^{r-1}$; plan families $(\mathcal P_r)$;
zero-error plan packets of lengths $n_r$ with $\alpha(G^{\boxtimes n_r})\ge|\mathcal P_r|$.
\State $r\gets 1$, $S_1\gets \mathcal K$, choose a hold arm $h\in\mathcal K$.
\While{$|S_r|>1$}
    \State \textbf{Install:} transmit a length-$n_r$ zero-error packet for $\text{plan}(S_r,m_r)$. \hspace{1em}(\emph{During installation:} agent pulls $h$.)
    \State \textbf{Execute:} for each $a\in S_r$, pull arm $a$ for $m_r$ rounds and \emph{count} these rewards.
    \State Let $t_r \gets 2^r-1$. For each $a\in S_r$, set
    $\mathrm{UCB}_a=\widehat\mu_a(t_r)+\beta(t_r)$, $\mathrm{LCB}_a=\widehat\mu_a(t_r)-\beta(t_r)$.
    \State $b_r\in\arg\max_{a\in S_r}\mathrm{LCB}_a$, $S_{r+1}\gets\{a\in S_r:\mathrm{UCB}_a \ge \mathrm{LCB}_{b_r}\}$;\quad $h \gets$  \emph{the last arm pulled in $r$}.
    \State $r\gets r+1$.
\EndWhile
\State \Return the unique arm in $S_r$.
\end{algorithmic}
\end{algorithm}

\paragraph{Confidence radius and phase schedule.}
Let 
\[
\beta(t)\triangleq \sqrt{\nicefrac{2\log\!\left(\frac{8K t^2}{\delta}\right)}{t}}.
\]
We use phases $r=1,2,\dots$ with per-phase budget $m_r\triangleq 2^{r-1}$ \emph{counted} pulls per active arm,
so that the cumulative number of counted pulls per surviving arm at the end of phase $r$ is
$t_r = \sum_{j=1}^r m_j = 2^r-1$.

During plan installation (the $n_r$ channel uses), the agent keeps executing its previously committed arm.

\paragraph{Plan packets.}
In phase $r$, the learner selects an active set $S_r\subseteq\mathcal K$ and a pre-agreed repetition budget
$m_r\in\mathbb N$. The phase plan is:
\[
\text{plan}(S_r,m_r): \ \text{pull each arm in $S_r$ exactly $m_r$ times.} 
\]
Let $\mathcal P_r$ denote the set of admissible phase plans (equivalently, admissible active sets).
A \emph{zero-error plan packet} of length $n_r$ is any code that can convey any index in $\mathcal P_r$ with
zero decoding error; its existence is ensured whenever $\alpha(G^{\boxtimes n_r})\ge |\mathcal P_r|$.

\begin{theorem}
\label{thm:PSE}
Assume $1$-subGaussian rewards, then PSE is $\delta$-correct. Moreover, with probability at least $1-\delta$, its total number
of physical rounds satisfies
\[
\tau \;\le\; N_{\rm SE}(\delta,\mu)\;+\;\sum_{r=1}^{R} n_r,
\]
where $R$ is the number of executed phases and $N_{\rm SE}(\delta,\mu)$ is the (instance-dependent) pull
complexity of the \emph{clean} phased successive-elimination algorithm (standard analyses), i.e.,
\[
N_{\rm SE}(\delta,\mu)
= \tilde O\!\left(\sum_{a\neq a^\star}\frac{1}{\Delta_a^2}\log\!\frac{K\log(1/\Delta_a)}{\delta}\right).
\]
\end{theorem}

\emph{Proof sketch.}
Define the standard uniform concentration event $
\mathcal E \triangleq \bigcap_{a\in\mathcal K}\bigcap_{t\ge 1}
\left\{|\widehat\mu_a(t)-\mu_a|\le \beta(t)\right\}.
$
By a union bound over arms and times, $\Pr(\mathcal E)\ge 1-\delta$. 
\newline
On $\mathcal E$, the elimination rule of phased successive elimination is sound: $a^\star$ is never removed, and any
suboptimal arm $a$ is eliminated once $4\beta(t)\le \Delta_a$ at the end of some phase (since
$\mathrm{UCB}_a\le \mu_a+2\beta(t)$ and $\mathrm{LCB}_{b_r}\ge \mathrm{LCB}_{a^\star}\ge \mu_{a^\star}-2\beta(t)$).
\newline
Crucially, PSE's \emph{counted} rewards are generated by exactly the same arm sequence and sample sizes as the clean
phased-SE algorithm (the additional $n_r$ rounds during plan installation are simply ignored for the statistical test).
Therefore, on $\mathcal E$, the number of counted pulls until termination is at most $N_{\rm SE}(\delta,\mu)$.
\newline
Finally, each executed phase incurs $n_r$ additional physical rounds for plan installation (switching latency), yielding
$\tau \le N_{\rm SE}(\delta,\mu)+\sum_{r=1}^R n_r$ on $\mathcal E$.\ifarxiv\ See full proof in App.~\ref{app:proof-PSE}.\fi


\begin{remark}
For $\beta(t)=\sqrt{2\log(8Kt^2/\delta)/t}$, the condition $4\beta(t)\le \Delta$ is implied by
$t \ge \frac{c}{\Delta^2}\log\!\Big(\frac{c'K}{\delta\Delta^2}\Big)$ for universal constants $c,c'>0$.
Consequently, $R^\star = 1+\lceil \log_2 T_{\max}\rceil = O(\log(1/\Delta_{\min}))$.
\end{remark}

\begin{corollary}
\label{cor:PSE-C5C6}
Consider the one-sided typewriter channel with $\varepsilon\in(0,1)$, so the confusability graph is $C_K$.
Assume the simple plan family $\mathcal P_r=\{\,\text{plan}(S,m_r):\emptyset\neq S\subseteq\mathcal K\,\}$,
so $|\mathcal P_r|\le 2^K-1$ for all $r$. 
\newline
\emph{(i) $K=5$.} Using the base-$5$ digit on $C_5$ (two uses per digit),
$n_r = 2\Big\lceil \log_5 |\mathcal P_r| \Big\rceil \le 2\Big\lceil \log_5(2^5) \Big\rceil = 6.$
\newline
\emph{(ii) $K=6$.} Using the calendar interface on $C_6$ (one ternary digit per use),
$n_r = \Big\lceil \log_3 |\mathcal P_r| \Big\rceil \le \Big\lceil \log_3(2^6) \Big\rceil = 4.$
\newline
In both cases, Theorem~\ref{thm:PSE} gives $\tau \le N_{\rm SE}(\delta,\mu) + O(R)$ with small constants:
at most $6R$ on $C_5$ and $4R$ on $C_6$.
\end{corollary}

Thus when the statistical term $N_{\rm SE}(\delta,\mu)$ dominates (small gaps / small $\delta$), PSE is dramatically better than per-pull zero-error updates, which incur a multiplicative constant $\approx n^\star(\mathcal{K})$.

\newpage
\IEEEtriggeratref{9}
\bibliographystyle{IEEEtran}
\bibliography{references}




\ifarxiv
\onecolumn 
\appendices 


\section{Proof of Proposition~\ref{lem:update-wrapper-modelS}}
\label{app:update-wrapper}

Fix a zero-error update code of blocklength $n_u$ for the message set $\mathcal K$:
an encoder $\mathrm{enc}:\mathcal K\to\mathcal X^{n_u}$ and decoder
$\mathrm{dec}:\mathcal Y^{n_u}\to\mathcal K$ such that for every $a\in\mathcal K$,
if $X^{n_u}=\mathrm{enc}(a)$ is transmitted then $\mathrm{dec}(Y^{n_u})=a$ almost surely.

We construct $\widetilde{\mathsf A}$ by simulating $\mathsf A$ as a \emph{virtual} clean algorithm.
Partition physical time into blocks $B_i\triangleq\{(i-1)n_u+1,\dots,in_u\}$ for $i=1,2,\dots$,
and let $t_i\triangleq in_u$ denote the last slot of block $i$.

At virtual step $i$, the simulated algorithm $\mathsf A$ outputs an arm request
$a_i\in\mathcal K$ as a measurable function of its past virtual rewards.
During physical block $B_i$, the learner transmits the length-$n_u$ codeword
$\mathrm{enc}(a_i)$ over the actuation channel.
By the assumed timing model, the agent executes the decoded arm \emph{immediately upon decoding},
i.e., in the last slot $t_i$ the executed arm equals $a_i$.
Define the virtual reward fed back to $\mathsf A$ at step $i$ to be the physical reward
$r_{t_i}$ observed in that last slot. All other rewards inside the block are ignored.

We now verify that the virtual interaction has the same law as a clean run of $\mathsf A$.
Condition on the virtual history up to step $i-1$, equivalently, on the $\sigma$-field generated by
$\{a_1,r_{t_1},\dots,a_{i-1},r_{t_{i-1}}\}$. Then the requested arm $a_i$ is determined.
By zero-error decoding and the “execute-on-decode” assumption, the executed arm at time $t_i$ equals $a_i$,
so the reward $r_{t_i}$ is distributed exactly as a clean reward sample from arm $a_i$
(and is independent of the past conditional on $a_i$ under the standard bandit model).
Hence the joint law of $(a_i,r_{t_i})_{i\ge 1}$ under $\widetilde{\mathsf A}$
matches the joint law of the arm/reward pairs produced by $\mathsf A$ in the clean bandit instance.
Therefore, $\widetilde{\mathsf A}$ is $\delta$-correct whenever $\mathsf A$ is $\delta$-correct.

Finally, if $\mathsf A$ stops after $\tau_{\rm clean}$ virtual pulls, then $\widetilde{\mathsf A}$
uses exactly $\tau = n_u\,\tau_{\rm clean}$ physical rounds (each virtual pull consumes one full block),
which proves the claimed time bound.

\section{Proof of Proposition~\ref{prop:C6-calendar-wrapper}}
\label{app:prop-C6-calendar-wrapper}
Let the public schedule activate $\mathcal S_{\rm even}$ on odd time slots and $\mathcal S_{\rm odd}$
on even time slots:
\[
s(t)=
\begin{cases}
\mathrm{even}, & t \text{ odd},\\
\mathrm{odd},  & t \text{ even}.
\end{cases}
\]
By Lemma~\ref{lem:oneuse-indset}, any transmitted arm in the active set is decoded with zero-error
in that same slot and is executed as intended.

We define a wrapper that simulates $\mathsf A$ as a virtual clean algorithm driven only by \emph{counted}
rewards. Let $a_1,a_2,\dots$ be the random sequence of arm requests produced by $\mathsf A$ when run
on the clean instance, and let $\tau_{\rm clean}$ be its (random) stopping time.
The wrapper maintains a virtual step counter $i$, initialized to $i=1$.
At each physical slot $t$:
\begin{itemize}[leftmargin=*]
\item If $a_i\in \mathcal S_{s(t)}$, i.e., the requested arm is admissible in the active parity class,
the learner transmits $a_i$, the agent executes $a_i$, the wrapper \emph{counts} the resulting reward
as the virtual reward for step $i$, and increments $i\leftarrow i+1$.
\item Otherwise, the learner transmits any arm in the active set, e.g., an arbitrary fixed element of
$\mathcal S_{s(t)}$, and discards the reward (this slot is an uncounted ``wait'' slot).
\end{itemize}
The wrapper terminates when $i=\tau_{\rm clean}+1$, i.e., once it has produced $\tau_{\rm clean}$ counted rewards.

\emph{Correctness \& coupling.}
Consider the subsequence of physical slots at which the wrapper counts rewards. 
By construction, at the $i$-th counted slot the executed arm equals $a_i$ (zero-error decoding and immediate execution).
Hence the $i$-th counted reward has exactly the same conditional distribution as a clean reward sample
from arm $a_i$. Since $\mathsf A$ chooses $a_i$ as a measurable function of past counted rewards,
the sequence of counted arm-reward pairs produced by the wrapper has the same joint distribution as the
clean interaction of $\mathsf A$. Therefore the wrapper preserves the error probability of $\mathsf A$,
and the resulting algorithm is $\delta$-correct.

\emph{Time bound.}

{Let $p_i \in \{E,O\}$ denote the parity class of the $i$-th arm
request $a_i$ made by the clean run of $\mathsf A$, where $E$ corresponds to
$\mathcal S_{\rm even}$ and $O$ to $\mathcal S_{\rm odd}$. Let $t_i$ be the
physical time at which the wrapper obtains the $i$-th counted reward.
Because the public schedule alternates between $\mathcal S_{\rm even}$ and
$\mathcal S_{\rm odd}$, and the wrapper always serves the current request at the
first slot in which its parity class is active, we have}
\[
{t_1 = 1 + \mathbf 1\{p_1 = O\},}
\]
{and for every $i \ge 2$,}
\[
{t_i = t_{i-1} + 1 + \mathbf 1\{p_i = p_{i-1}\}.}
\]
{Indeed, if $p_i \neq p_{i-1}$ then the next requested parity class is active
in the very next slot, whereas if $p_i = p_{i-1}$ the wrapper must wait one
additional slot before serving it. Summing the recurrence yields}
\[
{\tau = t_{\tau_{\rm clean}}
= \tau_{\rm clean}
+\mathbf 1\{p_1 = O\}
+\sum_{i=2}^{\tau_{\rm clean}}\mathbf 1\{p_i = p_{i-1}\}
\le 2\tau_{\rm clean}.}
\]

\section{Proof of Theorem~\ref{thm:PSE}}
\label{app:proof-PSE}

We recall that rewards are independent across rounds, and when arm $a\in\mathcal K$ is executed
the reward distribution has mean $\mu_a$ and is $1$-subGaussian. The confidence radius is
\[
\beta(t)\triangleq \sqrt{\frac{2\log\!\left(\frac{8K t^2}{\delta}\right)}{t}},\qquad t\ge 1.
\]
For PSE, in phase $r$ the per-arm \emph{counted} budget is $m_r=2^{r-1}$, hence each arm in $S_r$
has exactly
\[
t_r \triangleq \sum_{j=1}^r m_j = 2^r-1
\]
counted samples at the end of phase $r$.

\subsection{Step 1: Anytime concentration with explicit constants}

For each arm $a$, let $\widehat\mu_a(t)$ denote the empirical mean of the first $t$ \emph{counted} rewards
obtained when executing arm $a$. (Equivalently, one may imagine an infinite i.i.d.\ sequence of arm-$a$ rewards,
and $\widehat\mu_a(t)$ is the mean of the first $t$ draws; this is standard and matches the algorithm because
PSE only ever uses $\widehat\mu_a(t)$ at times $t=t_r$ for arms still active.)

\begin{lemma}
\label{lem:anytime-conc}
Define the event
\[
\mathcal E \triangleq \bigcap_{a\in\mathcal K}\bigcap_{t\ge 1}
\left\{\,|\widehat\mu_a(t)-\mu_a|\le \beta(t)\right\}.
\]
If rewards are $1$-subGaussian, then $\Pr(\mathcal E)\ge 1-\delta$.
\end{lemma}

\emph{Proof:}
Fix an arm $a$ and a time $t\ge 1$. Since rewards are $1$-subGaussian, the empirical mean satisfies
\[
\Pr\!\left(\left|\widehat\mu_a(t)-\mu_a\right|>\varepsilon\right)\le 2\exp\!\left(-\frac{t\varepsilon^2}{2}\right)
\qquad\forall \varepsilon>0.
\]
Plugging in $\varepsilon=\beta(t)$ yields
\begin{align*}
\Pr\!\left(\left|\widehat\mu_a(t)-\mu_a\right|>\beta(t)\right)
&\le 2\exp\!\left(-\frac{t\beta(t)^2}{2}\right)\\
&= 2\exp\!\left(-\log\!\left(\frac{8Kt^2}{\delta}\right)\right)\\
&= \frac{\delta}{4K t^2}.
\end{align*}
By a union bound over all $a\in\mathcal K$ and all $t\ge 1$,
\[
\Pr(\mathcal E^c)
\le \sum_{a\in\mathcal K}\sum_{t=1}^\infty \frac{\delta}{4K t^2}
= \frac{\delta}{4}\sum_{t=1}^\infty \frac{1}{t^2}
\le \frac{\delta}{4}\cdot \frac{\pi^2}{6}
< \delta,
\]
hence $\Pr(\mathcal E)\ge 1-\delta$.

\subsection{Step 2: Exact coupling of counted samples}

We now state the coupling lemma formally. Define the \emph{clean} phased SE algorithm as the same phased
procedure as PSE but with no installation delays: in phase $r$, it pulls each arm in $S_r$ exactly $m_r$
times (counting all rewards), then applies the same elimination rule as PSE using $\beta(t_r)$, and repeats
until one arm remains.

\begin{lemma}
\label{lem:counted-coupling}
Assume every phase-$r$ plan packet in PSE is decoded with zero error, so the agent executes the intended
phase plan $\text{\rm plan}(S_r,m_r)$ exactly.
Let $(A^{\rm cnt}_s,R^{\rm cnt}_s)_{s\ge 1}$ denote the sequence of \emph{counted} arm pulls and counted rewards
produced by PSE, indexed by counted time $s=1,2,\dots$, i.e., ignoring all installation rounds.
Let $(A^{\rm clean}_s,R^{\rm clean}_s)_{s\ge 1}$ denote the pull, or equivalently reward, sequence of the clean phased-SE algorithm
run with the same $\beta(\cdot)$, the same phase budgets $m_r$, and the same elimination rule.

Then there exists a coupling under which
\[
(A^{\rm cnt}_s,R^{\rm cnt}_s)_{s\ge 1}\equiv (A^{\rm clean}_s,R^{\rm clean}_s)_{s\ge 1}
\qquad\text{almost surely}.
\]
In particular, PSE and clean phased-SE compute identical empirical means from counted samples, hence produce the same
active sets $(S_r)$ and stop after the same number of counted pulls.
\end{lemma}

\emph{Proof:}
Construct a probability space as follows. For each arm $a\in\mathcal K$, generate two independent i.i.d.\ sequences:
\[
\left(R^{(c)}_{a,1},R^{(c)}_{a,2},\dots\right)\quad\text{and}\quad
\left(R^{(u)}_{a,1},R^{(u)}_{a,2},\dots\right),
\]
each distributed as the arm-$a$ reward law (mean $\mu_a$, $1$-subGaussian).
Interpret $R^{(c)}$ as the rewards that will be used for \emph{counted} pulls of arm $a$, and $R^{(u)}$ as rewards
used for \emph{uncounted} pulls, i.e., installation rounds. This is valid because, in the true system, every time arm $a$
is executed the reward is an independent draw from the same law; splitting draws into two independent pools
preserves the joint law of any finite collection of executed rewards.

Now run the clean phased-SE algorithm and define that its $j$-th pull of arm $a$ receives reward $R^{(c)}_{a,j}$.
Run PSE on the same probability space, with the following reward assignment:
\begin{itemize}[leftmargin=*]
\item During \emph{execution} (counted) rounds, whenever PSE executes arm $a$ for the $j$-th counted time,
it receives reward $R^{(c)}_{a,j}$.
\item During \emph{installation} (uncounted) rounds, PSE executes the hold arm $h$ (the last arm pulled
in the previous phase) and receives rewards from the $R^{(u)}$ pool (which are never used in estimates).
\end{itemize}

Because plan packets decode with zero error, in each phase $r$, PSE executes exactly the intended plan
$\text{plan}(S_r,m_r)$ in its execution segment. Therefore, the counted arm sequence of PSE in phase $r$ is exactly
the same as the arm sequence of the clean phased-SE algorithm in phase $r$ (both pull each $a\in S_r$ exactly $m_r$
times). Under the construction above, they also receive exactly the same counted rewards, namely the corresponding
$R^{(c)}$ samples. Hence, after phase $r$, both algorithms compute identical empirical means
$\widehat\mu_a(t_r)$ for each $a\in S_r$, apply the same elimination rule, and therefore produce the same $S_{r+1}$.

By induction over phases, the entire counted pull/reward sequence agrees almost surely under this coupling.

\subsection{Step 3: Correctness of elimination on the good event}

\begin{lemma}
\label{lem:elim-sound}
On the event $\mathcal E$ from Lemma~\ref{lem:anytime-conc}, PSE never eliminates the best arm $a^\star$.
Moreover, fix any suboptimal $a\neq a^\star$. For any phase index $r$ such that
\begin{equation}
\label{eq:elim-condition-4beta}
4\beta(t_r) < \Delta_a,
\end{equation}
arm $a$ is eliminated at the end of phase $r$ (i.e., $a\notin S_{r+1}$).
\end{lemma}

\emph{Proof:}
Fix a phase $r$ and let $t_r=2^r-1$ be the counted sample size per surviving arm at the end of phase $r$.

\emph{(i) The best arm is not eliminated.}
Let $b_r\in\arg\max_{b\in S_r}\mathrm{LCB}_b$ where
\[
\mathrm{LCB}_b \triangleq \widehat\mu_b(t_r)-\beta(t_r),\qquad
\mathrm{UCB}_b \triangleq \widehat\mu_b(t_r)+\beta(t_r).
\]
On $\mathcal E$, we have $\widehat\mu_{a^\star}(t_r)\ge \mu_{a^\star}-\beta(t_r)$, hence
\[
\mathrm{UCB}_{a^\star}=\widehat\mu_{a^\star}(t_r)+\beta(t_r)\ge \mu_{a^\star}.
\]
Also on $\mathcal E$, for any $b\in S_r$ we have $\widehat\mu_b(t_r)\le \mu_b+\beta(t_r)$, so
\[
\mathrm{LCB}_{b_r}=\widehat\mu_{b_r}(t_r)-\beta(t_r)\le \mu_{b_r}\le \mu_{a^\star}.
\]
Therefore $\mathrm{UCB}_{a^\star}\ge \mu_{a^\star}\ge \mathrm{LCB}_{b_r}$, so $a^\star$ satisfies the retention
condition $\mathrm{UCB}_{a^\star}\ge \mathrm{LCB}_{b_r}$ and is not eliminated.

\emph{(ii) A suboptimal arm is eliminated once $4\beta(t_r)<\Delta_a$.}
Fix $a\neq a^\star$. On $\mathcal E$ we have
\[
\mathrm{UCB}_a
=\widehat\mu_a(t_r)+\beta(t_r)
\le (\mu_a+\beta(t_r))+\beta(t_r)=\mu_a+2\beta(t_r).
\]
Also, since $b_r$ maximizes $\mathrm{LCB}$ over $S_r$ and $a^\star\in S_r$ by part (i),
\[
\mathrm{LCB}_{b_r}\ge \mathrm{LCB}_{a^\star}
= \widehat\mu_{a^\star}(t_r)-\beta(t_r)\ge (\mu_{a^\star}-\beta(t_r))-\beta(t_r)=\mu_{a^\star}-2\beta(t_r).
\]
Combining the two displays,
\[
\mathrm{UCB}_a \le \mu_a+2\beta(t_r)=\mu_{a^\star}-\Delta_a+2\beta(t_r).
\]
If $4\beta(t_r)<\Delta_a$, then $-\Delta_a+2\beta(t_r)<-2\beta(t_r)$, hence
\[
\mathrm{UCB}_a < \mu_{a^\star}-2\beta(t_r)\le \mathrm{LCB}_{b_r}.
\]
Therefore $a$ fails the retention test $\mathrm{UCB}_a\ge \mathrm{LCB}_{b_r}$ and is eliminated, i.e.,
$a\notin S_{r+1}$.

\subsection{Step 4: Physical stopping time decomposition and conclusion}

We now prove Theorem~\ref{thm:PSE}.

\emph{Proof:}
Assume plan packets are decoded with zero error (as in the theorem statement). Let $R$ be the number of executed phases, i.e., the number of iterations of the while-loop in Algorithm~\ref{alg:PSE}.
Define the number of \emph{counted} pulls made by PSE as
\[
N_{\rm counted}\triangleq \sum_{r=1}^{R} |S_r|\,m_r.
\]
In each executed phase $r$, PSE spends exactly $n_r$ physical rounds to transmit the plan packet (during which it pulls
the hold arm and discards these rewards for estimation), and then spends exactly $|S_r|m_r$ physical rounds executing
the plan (counting those rewards). Therefore the total number of physical rounds is exactly
\begin{equation}
\label{eq:tau-decomposition}
\tau = \sum_{r=1}^{R} \big(n_r+|S_r|m_r\big)=N_{\rm counted}+\sum_{r=1}^{R} n_r.
\end{equation}

Next, by Lemma~\ref{lem:counted-coupling}, the counted pull/reward sequence of PSE is (under a coupling) identical to
that of the clean phased-SE algorithm with the same $(m_r)$ and elimination rule. In particular, PSE and clean phased-SE
stop after the same number of counted pulls and output the same arm as a function of the counted data. Let
$\tau_{\rm SE}$ denote the (random) stopping time (number of pulls) of this clean phased-SE algorithm. Then
\[
N_{\rm counted}=\tau_{\rm SE}\qquad\text{almost surely under the coupling},
\]
hence, in distribution as well.

By Lemma~\ref{lem:anytime-conc}, $\Pr(\mathcal E)\ge 1-\delta$. On $\mathcal E$, Lemma~\ref{lem:elim-sound} implies that
$a^\star$ is never eliminated and every suboptimal arm is eventually eliminated, so both algorithms terminate and output
$a^\star$ on $\mathcal E$. Therefore PSE is $\delta$-correct.

Finally, standard analyses of phased successive elimination imply that with probability at least $1-\delta$,
\[
\tau_{\rm SE} \le N_{\rm SE}(\delta,\mu),
\]
where $N_{\rm SE}(\delta,\mu)$ is the usual instance-dependent clean sample complexity bound (one such bound is stated
in Theorem~\ref{thm:PSE}). Combining with \eqref{eq:tau-decomposition} yields, with probability at least $1-\delta$,
\[
\tau = N_{\rm counted}+\sum_{r=1}^{R}n_r = \tau_{\rm SE}+\sum_{r=1}^{R} n_r
\le N_{\rm SE}(\delta,\mu)+\sum_{r=1}^{R} n_r,
\]
which concludes our proof.

\section{Missing Proof Sketches}
This section contains the proof (sketches) we removed from the main body of the paper.

\subsection{Proof Sketch of Lemma \ref{lem:baseline-nonident}}
If $W\mu=W\mu'$, then for every command $i$ the conditional reward law given $X_t=i$ is identical under $\mu$
and $\mu'$. Hence the entire observation process has the same distribution under both instances, so any
decision rule must behave identically and cannot be correct on both.

\subsection{Proof of Lemma \ref{lem:indset-decoding}}
By definition of the confusability graph, $\{x,x'\}\in E$ iff
$\mathcal{Y}(x)\cap\mathcal{Y}(x')\neq\emptyset$.
Independence of $\mathcal S$ means that for any distinct $x,x'\in\mathcal S$ we have
$\mathcal{Y}(x)\cap\mathcal{Y}(x')=\emptyset$. Therefore, any output $y$ can come from
\emph{at most one} $x\in\mathcal S$, and the decoder can recover $X$ as the unique
element of $\{x\in\mathcal S:W(y\mid x)>0\}$.

\subsection{Proof Sketch of Lemma \ref{lem:oneuse-indset}}
For the typewriter channel, $y\in\{x,x+1\}$. Since $\mathcal S$ is independent, $x+1\notin\mathcal S$.
Thus, the decoder rule
$\hat x(y\mid\mathcal S)=y$ if $y\in\mathcal S$ and $\hat x(y\mid\mathcal S)=y-1\pmod K$ otherwise
returns $x$ in both cases.

\fi 

\end{document}